\documentclass[11pt,a4paper]{article}

\usepackage[utf8]{inputenc}
\usepackage[T1]{fontenc}
\usepackage[margin=1in]{geometry}
\usepackage{amsmath,amssymb,amsthm}
\usepackage{graphicx}
\usepackage{booktabs}
\usepackage{hyperref}
\usepackage{xcolor}
\usepackage{enumitem}
\usepackage{algorithm}
\usepackage{algpseudocode}
\usepackage{caption}
\usepackage{tabularx}
\usepackage{url}
\usepackage{tikz}
\usetikzlibrary{arrows.meta,positioning,shapes.geometric,fit,calc}

\newtheorem{definition}{Definition}[section]
\newtheorem{axiom}{Axiom}[section]
\newtheorem{invariant}{Invariant}[section]

\newtheorem{corollary}{Corollary}[section]
\newtheorem{remark}{Remark}[section]

\newcommand{\PunkGo}{\textsc{PunkGo}}

\newcommand{\actors}{\ensuremath{\mathcal{P}}}
\newcommand{\actorsh}{\ensuremath{\mathcal{P}_H}}
\newcommand{\actorsa}{\ensuremath{\mathcal{P}_A}}
\newcommand{\targets}{\ensuremath{\mathcal{T}}}
\newcommand{\logL}{\ensuremath{\mathcal{L}}}

\hypersetup{
  colorlinks=true,
  linkcolor=blue!70!black,
  citecolor=blue!70!black,
  urlcolor=blue!70!black
}

\title{\textbf{Right to History: A Sovereignty Kernel\\for Verifiable AI Agent Execution}}

\author{
  Jing Zhang\\
  Independent Researcher\\
  PunkGo Project (\url{https://punkgo.ai})\\
  \texttt{feijiu@punkgo.ai}
}

\date{February 2026}

\begin{document}
\maketitle

\begin{abstract}
AI agents increasingly act on behalf of humans, yet no existing system provides a tamper-evident, independently verifiable record of what they did. As regulations such as the EU AI Act begin mandating automatic logging for high-risk AI systems, this gap carries concrete consequences---especially for agents running on personal hardware, where no centralized provider controls the log. Extending Floridi's informational rights framework from data about individuals to actions performed on their behalf, this paper proposes the \emph{Right to History}: the principle that individuals are entitled to a complete, verifiable record of every AI agent action on their own hardware. The paper formalizes this principle through five system invariants with structured proof sketches, and implements it in \PunkGo{}, a Rust sovereignty kernel that unifies RFC~6962 Merkle tree audit logs, capability-based isolation, energy-budget governance, and a human-approval mechanism. Adversarial testing confirms all five invariants hold. Performance evaluation shows sub-1.3\,ms median action latency, $\sim$400 actions/sec throughput, and 448-byte Merkle inclusion proofs at 10{,}000 log entries.
\end{abstract}

\section{Introduction}
\label{sec:intro}

AI agents are executing an expanding range of actions on behalf of humans: sending emails, modifying files, running code, calling external APIs. Yet no system today can produce a cryptographic proof of what was done, who authorized it, and whether the agent stayed within its declared bounds. OS-level audit mechanisms (auditd, eBPF tracing) operate at syscall granularity and cannot attribute actions to agent-level semantics or support human-in-the-loop governance; transparent log systems (Certificate Transparency~\cite{rfc6962}, Trillian~\cite{trillian}) provide the cryptographic tooling but have never been applied to agent action recording. As regulations such as the EU AI Act~\cite{eu-ai-act} and the NIST AI Risk Management Framework~\cite{nist-ai-rmf} begin mandating automatic logging and transparency for high-risk AI systems, this absence of verifiable agent history is no longer merely a design gap---it is a regulatory liability.

The problem is especially acute when AI agents run on personal hardware. Current regulatory frameworks implicitly assume a centralized provider controlling a log server; when users run agents on their own machines, there is no cloud, no provider, no centralized log. Who guarantees that records exist, are complete, and have not been tampered with?

This paper argues that this constitutes a missing right. Floridi's information ethics framework~\cite{floridi-ethics} establishes that informational entities possess rights related to the creation, processing, and use of information about them---rights that existing legal instruments have begun to codify through data access and transparency obligations~\cite{gdpr,eu-ai-act}. This paper extends this line of reasoning from \emph{data about individuals} to \emph{computational actions performed on their behalf}: if individuals have a right to access data held about them, they should equally have a right to access a verifiable record of what AI agents \emph{did} on their behalf. The \emph{Right to History} is the principle asserting that individuals are entitled, on their own hardware, to a complete, verifiable, tamper-evident record of every AI agent action---not an optional feature but an architectural commitment that shapes the system's structure.

Existing work touches different facets of this problem but none provides a complete answer. Agent Contracts~\cite{agent-contracts} formalize what an agent \emph{is allowed} to do---but do not record what it \emph{actually did}. AIOS~\cite{aios} builds an agent operating system---but places the Large Language Model (LLM), rather than a verified kernel, as the trusted computing base (TCB), and provides no audit log. Certificate Transparency~\cite{rfc6962} demonstrates that Merkle tree append-only logs can protect billions of TLS certificates---but has never been applied to AI agent actions. seL4~\cite{sel4} proves that an OS kernel can formally verify capability isolation---but does not record what happens within those isolation boundaries.

This paper presents \PunkGo{}, a \emph{sovereignty kernel} that fills the gap between these bodies of work. The contributions are:

\begin{description}[leftmargin=1.5em,labelindent=0em]
\item[C1. Formal Framework.] A formal model for AI agent action auditing comprising three primitives (Actor, Energy, State), four atomic actions (\texttt{observe}, \texttt{create}, \texttt{mutate}, \texttt{execute}), and five system invariants (Append-Only, Completeness, Integrity, Boundary Enforcement, Energy Conservation), with proof sketches for each invariant following a structured pre-/post-condition argument style. Unlike Agent Contracts~\cite{agent-contracts}, which specify what agents may do, this framework argues that what agents actually did is completely recorded and tamper-evident.

\item[C2. PunkGo Kernel.] A Rust implementation that unifies RFC~6962-style Merkle tree audit logs, capability-based writable-target boundaries, and energy budget governance (including a \texttt{hold\_on} human-approval mechanism) within a kernel-level TCB. Unlike AIOS~\cite{aios}, which places the LLM as TCB, \PunkGo{} places a \emph{verified kernel} as TCB.

\item[C3. Sovereignty Inversion.] A directional asymmetry in the AI sovereignty literature: existing work asks how agents maintain sovereignty \emph{from} humans~\cite{sovereign-agents}; the complementary question---how humans maintain sovereignty \emph{over} agents---has received less attention. The Right to History addresses the latter.

\item[C4. Proof-of-Concept Validation.] Invariant satisfiability is verified on the implemented kernel and key performance characteristics are measured: sub-1.3\,ms median action latency, $\sim$400 actions/sec throughput, and 448-byte Merkle inclusion proofs at 10{,}000 log entries.
\end{description}

The remainder of this paper is organized as follows. \S\ref{sec:background} provides background and motivation. \S\ref{sec:threat} defines the threat model. \S\ref{sec:formal} presents the formal framework (core contribution). \S\ref{sec:design} describes the system design. \S\ref{sec:eval} reports evaluation results. \S\ref{sec:related} discusses related work. \S\ref{sec:discussion} addresses limitations and future work. \S\ref{sec:conclusion} concludes.

\section{Background and Motivation}
\label{sec:background}

\subsection{The Accountability Gap in the Age of AI Agents}

AI agents are transitioning from experimental tools to everyday executors. Users delegate file management, communications, external service calls, and even arbitrary code execution to LLM-driven agents. This delegation creates a fundamental accountability problem: when actions are driven by an unpredictable LLM, how can humans always (a)~know what happened, (b)~prove to third parties what happened, and (c)~verify that actions stayed within authorized bounds?

Existing agent frameworks fail to answer this question. AIOS~\cite{aios} provides agent scheduling and context management, but its access control is a rudimentary permission-group hashmap with no audit log and no cryptographic evidence. LangGraph offers an \texttt{interrupt()} pause mechanism that produces no verifiable history. The OpenAI SDK's \texttt{needsApproval} flag is a boolean switch with no economic constraint against abuse.

The pattern is clear: \emph{academia and industry have built the policy layer but not the evidence layer}. Agent Contracts~\cite{agent-contracts} define bounds, AgentSentry~\cite{agentsentry} generates ephemeral permissions, Governance-as-a-Service~\cite{governanceaas} enforces runtime policies---but none of these systems produce a persistent, tamper-evident, independently verifiable action record.

\subsection{The Transparent Log Paradigm}

The cryptographic community has solved a structurally analogous problem. Certificate Transparency (CT)~\cite{rfc6962} uses Merkle tree append-only logs to protect the TLS certificate ecosystem: every certificate issuance is recorded in a public log, inclusion proofs are $O(\log n)$, and consistency proofs guarantee that the log has not forked. CT now protects billions of certificates, and its Merkle tree structure has been extended to Key Transparency (CONIKS~\cite{coniks}), Go module mirrors (Trillian), and general-purpose verifiable data structures.

Crosby and Wallach~\cite{crosby-wallach} formalized data structures for tamper-evident logging, demonstrating that a log of 80~million events requires only 3\,KB of proof. Chun et~al.'s Attested Append-Only Memory (A2M)~\cite{a2m} goes further, embedding an append-only primitive in the TCB, allowing PBFT to tolerate $1/2$ Byzantine faults (the classical bound is $1/3$) with only $\sim$4\% overhead.

These works establish the paradigm's feasibility and efficiency. \emph{Yet no one has applied the transparent log paradigm to OS-level AI agent action auditing.} CT records certificates; CONIKS records keys; \PunkGo{} records agent actions.

\subsection{Regulatory Urgency}

The EU AI Act~\cite{eu-ai-act} creates mandatory demand for auditable AI infrastructure. Article~12 requires that high-risk AI systems ``shall technically allow for the automatic recording of events (`logs')''; Article~14 requires human oversight capability, including the ability to intervene in or interrupt operation; Article~19 specifies a minimum six-month log retention period. Penalties for non-compliance reach 7\% of global annual turnover or EUR~35~million. High-risk system requirements become enforceable on August~2, 2026.

However, Article~12 implicitly assumes a centralized provider scenario. For AI agents running on personal hardware---a rapidly growing use case---there is no provider controlling the log. The user simultaneously serves as provider, deployer, and overseer. \PunkGo{} enables individuals to self-comply with Articles~12 and~14, reframing logging from a compliance burden into a user right.

\subsection{From Data Sovereignty to Computational History Sovereignty}

Existing work on data sovereignty~\cite{local-first,solid} establishes principles for user data ownership---but does not address \emph{computational history ownership}. You can own your data, but can you prove what an agent \emph{did} with your data? The Right to History (\S\ref{sec:intro}) fills this gap: extending the trajectory from data ownership to action-record ownership. You need not control every agent decision (that is the enemy of efficiency), but you should be entitled to a complete, verifiable, tamper-evident record of every agent action on your hardware.

\section{Threat Model and Assumptions}
\label{sec:threat}

\subsection{System Boundary and Trust Model}

\PunkGo{}'s trust model places the kernel as the TCB. All other components---including AI agents (LLMs), external tools, and the communication layer---are treated as untrusted.

\paragraph{Inside the TCB (trusted).} The \PunkGo{} Kernel: event log (Merkle tree), energy ledger, boundary checker, payload validator, and audit prover.

\paragraph{Outside the TCB (untrusted).} AI agents (LLM-driven), external tools, client interfaces (remote consoles), and the network communication layer.

\paragraph{Trust Assumptions.}

\begin{description}[leftmargin=1.5em,labelindent=0em]
\item[A1 (Kernel Trusted).] The kernel code correctly implements the protocol described in this paper. This is the standard TCB assumption, consistent with the pre-verification premise of seL4~\cite{sel4}. Rust's memory safety guarantees eliminate entire vulnerability classes (buffer overflows, use-after-free) without runtime overhead; the implementation uses zero \texttt{unsafe} blocks in the core pipeline. Machine-checked verification (using Verus~\cite{verus}) is future work.

\item[A2 (Hardware Trusted).] The underlying processor and storage correctly execute instructions and persist data. This assumption can be further strengthened via TEEs (e.g., Intel SGX/TDX), listed as future work.

\item[A3 (Agent Untrusted).] This is the core design premise. LLM-driven agents are non-deterministic, not formally verifiable, and may produce arbitrary output. The entire system design revolves around this point: the design does not trust agent behavior, but trusts the kernel's record of agent behavior.

\item[A4 (Cryptographic Assumption).] The hash function used (SHA-256) is collision-resistant.
\end{description}

\subsection{Adversary Model}

The threat model considers five attack scenarios:

\begin{description}[leftmargin=1.5em,labelindent=0em]
\item[T1 (Privilege Escalation).] An agent attempts to operate on targets outside its declared \texttt{writable\_targets}---for example, a read-only agent attempting a write.

\item[T2 (Budget Exhaustion).] An agent repeatedly triggers \texttt{hold\_on} approvals to lock energy, starving other operations of resources.

\item[T3 (History Tampering).] A compromised non-kernel component (Level~1 adversary) attempts to modify or delete committed events in the log via the kernel API. Direct database manipulation by a root-privileged adversary (Level~2) is outside the current threat model.

\item[T4 (Equivocation / Log Forking).] A compromised component attempts to present different versions of history to different verifiers. Full defense requires external root comparison (see Scope of Defense below and Corollary~4.1).

\item[T5 (Approval Bypass).] An agent attempts to execute high-risk operations without triggering \texttt{hold\_on} checks.
\end{description}

\paragraph{Scope of Defense.}
It is essential to distinguish two levels of adversary:

\emph{Level~1: Untrusted agents (the primary threat).} Agents operate as unprivileged actors within the kernel's jurisdiction. The five invariants (INV-1--5) provide full defense at this level: agents cannot escape capability boundaries, exhaust energy beyond allocation, tamper with committed logs, or bypass human approval. All guarantees in this paper apply to Level~1 adversaries.

\emph{Level~2: Privileged platform adversary.} A user (or attacker) with root access to the host OS can stop the kernel process, directly modify the underlying storage, and reconstruct the Merkle tree. \PunkGo{} does not currently defend against Level~2 adversaries. Assumptions A1 (kernel trusted) and A2 (hardware trusted) explicitly exclude this case. Defense at Level~2 requires hardware-rooted trust: running the kernel inside a Trusted Execution Environment (TEE) such as Intel SGX/TDX, combined with remote attestation and external root anchoring (analogous to Certificate Transparency's gossip protocol). TEE integration is listed as future work (\S\ref{sec:discussion}).

This two-level distinction is consistent with systems security practice: seL4~\cite{sel4} proves isolation properties under analogous TCB assumptions without defending against physical hardware attacks; Certificate Transparency~\cite{rfc6962} assumes log servers are honest but detects equivocation when they are not, via external gossip---a mechanism \PunkGo{} does not yet implement but whose architectural slot (the Audit Prover's exported roots) is already in place.

\subsection{Security Goals}

Table~\ref{tab:security-goals} maps each security goal to the threats it addresses and the invariant that provides the guarantee.

\begin{table}[t]
\centering
\caption{Security goals, threats, and corresponding invariants.}
\label{tab:security-goals}
\small
\begin{tabularx}{\columnwidth}{lXcc}
\toprule
\textbf{Goal} & \textbf{Description} & \textbf{Threat} & \textbf{Inv.} \\
\midrule
G1 & Every capability-gated action is recorded & T3 & INV-2 \\
G2 & Committed events cannot be modified or deleted & T3 & INV-1 \\
G3 & Third parties can independently verify event existence and log consistency & T4 & INV-3 \\
G4 & Agents cannot act beyond declared \texttt{writable\_targets} & T1, T5 & INV-4 \\
G5 & Agents cannot consume energy beyond their allocation & T2 & INV-5 \\
\bottomrule
\end{tabularx}
\end{table}

\section{Formal Framework}
\label{sec:formal}

This section presents the core theoretical contribution. It first defines the system model (\S\ref{sec:world}), then formalizes the event log (\S\ref{sec:log}), capability model (\S\ref{sec:cap}), energy model (\S\ref{sec:energy}), and hold mechanism (\S\ref{sec:hold}), and finally states and argues five system invariants (\S\ref{sec:invariants}).

\subsection{System Model}
\label{sec:world}

\begin{definition}[World]
A \PunkGo{} system is a triple $W = (\actors, E, S)$, where $\actors$ is a finite set of actors, $E: \actors \to \mathbb{R}_{\geq 0}$ is the energy function, and $S = (\logL, W, \mathcal{V})$ is the system state comprising the event log $\logL$, the writability map $W: \actors \to 2^{\mathit{Pattern} \times \mathit{ActionType}}$, and a set of active envelopes $\mathcal{V}$.
\end{definition}

\begin{definition}[Actor Partition]
$\actors = \actorsh \cup \actorsa$, where $\actorsh$ is the set of human actors and $\actorsa$ is the set of agent actors. $\actorsh \cap \actorsa = \emptyset$ (exhaustive and mutually exclusive).
\end{definition}

\begin{axiom}[Human Unconditional Existence]
$\forall\, p \in \actorsh$: $p$ cannot be deleted or demoted.
\end{axiom}

\begin{axiom}[Agent Conditional Existence]
$\forall\, a \in \actorsa$: $a$ must declare $\mathrm{creator}(a) \in \actors$ and $\mathrm{purpose}(a)$ at creation time. $\mathrm{creator}(a)$ may set $\mathrm{expiry}(a)$.
\end{axiom}

\begin{definition}[Atomic Action]
An action $A$ is a 5-tuple $A = (\mathit{actor}, \mathit{type}, \mathit{target}, \mathit{payload}, \mathit{timestamp})$, where $\mathit{actor} \in \actors$, $\mathit{type} \in \{\texttt{observe}, \texttt{create}, \texttt{mutate}, \texttt{execute}\}$, $\mathit{target} \in \targets$ (the target object identifier), $\mathit{payload}$ is the action parameter, and $\mathit{timestamp} \in \mathbb{N}$ (nanoseconds since epoch; the implementation uses 64-bit unsigned integers).
\end{definition}

\begin{definition}[State Transition]
The system advances by actions, not by time:
\[
S' = f(S, A, E) \qquad E' = E - \mathrm{cost}(A)
\]
where $f$ is a deterministic state transition function and $\mathrm{cost}: \mathit{ActionType} \to \mathbb{R}_{\geq 0}$.
\end{definition}

\begin{definition}[Sole Committer]
\label{def:committer}
There exists a unique Committer role $\kappa$ responsible for transforming actions into committed events. $\kappa$ is a structural role, not a moral authority.
\end{definition}

\begin{remark}
The sole-committer design is deliberate: single-point linearization yields a deterministic total-order log. This aligns with Lamport's total-order broadcast~\cite{lamport} but simplifies to in-process serialization in the single-machine setting. Multi-committer extension (requiring consensus) is future work.
\end{remark}

\subsection{Event Log}
\label{sec:log}

\begin{definition}[Event]
An event $e$ is a tuple $e = (\mathit{id}, \mathit{seq}, \mathit{actor}, \mathit{type}, \mathit{target}, \mathit{payload}, \mathit{payload\_hash}, \mathit{timestamp}, \mathit{artifact\_hash}, \mathit{reserved\_energy}, \mathit{settled\_energy}, \mathit{event\_hash})$, where $\mathit{id}$ is a unique identifier (UUID), $\mathit{seq} \in \mathbb{N}$ is a monotonically increasing sequence number, $\mathit{payload\_hash} = H(\mathrm{canonical}(\mathit{payload}))$, $\mathit{artifact\_hash}$ is present for \texttt{execute} actions (echoing the submitted hash), $\mathit{reserved\_energy}$ and $\mathit{settled\_energy}$ record the energy accounting, and $\mathit{event\_hash} = H(0\mathtt{x00} \| \mathrm{canonical}(e))$ is the RFC~6962 leaf hash.
\end{definition}

\begin{remark}
Inter-event integrity is maintained through the Merkle tree structure rather than per-event hash chains. This is a design choice: Merkle trees provide $O(\log n)$ inclusion and consistency proofs, whereas per-event hash chains support only $O(n)$ prefix verification. \PunkGo{} uses the Google tlog algorithm (C2SP tlog-checkpoint format), the same structure used by Certificate Transparency and the Go Module Mirror.
\end{remark}

\begin{definition}[Log]
A log $\logL = \langle e_1, e_2, \ldots, e_n \rangle$ is a finite totally ordered sequence of events satisfying: (i)~$\forall\, i: e_i.\mathit{seq} = i$ (contiguous sequence numbers), and (ii)~Merkle tree leaves are ordered by $\mathit{seq}$.
\end{definition}

\begin{definition}[Merkle Root]
For a log $\logL$, its Merkle root is $\mathrm{root}(\logL) = \mathrm{MerkleRoot}(e_1, \ldots, e_n)$, following the RFC~6962 Merkle Hash Tree construction~\cite{rfc6962}.
\end{definition}

\begin{definition}[Inclusion Proof]
For event $e_i \in \logL$, the inclusion proof $\pi(e_i, \logL)$ is a Merkle path from $e_i$ to $\mathrm{root}(\logL)$, of size $O(\log |\logL|)$. A verifier holding $\mathrm{root}(\logL)$ and $\pi$ can verify the existence of $e_i$.
\end{definition}

\begin{definition}[Consistency Proof]
For two log versions $\logL$ and $\logL'$ ($|\logL| \leq |\logL'|$), the consistency proof $\sigma(\logL, \logL')$ demonstrates that $\logL$ is a prefix of $\logL'$, with size $O(\log |\logL'|)$.
\end{definition}

\subsection{Capability Model}
\label{sec:cap}

The capability model follows the principle of least privilege, drawing on the foundational capability concept introduced by Dennis and Van Horn~\cite{dennis-vanhorn}.

\begin{definition}[Writable Declaration]
Each actor $p \in \actors$ has a writability set $W(p) \subseteq \mathit{Pattern} \times \mathit{ActionType}$, where $\mathit{Pattern}$ is a glob-pattern over target identifiers. $W(p)$ is declared at the time $p$ is created.
\end{definition}

\begin{axiom}[Default Deny]
\label{ax:deny}
$\forall$ action $A$ by actor $p$: if $\nexists\, (\mathit{pat}, \mathit{type}) \in W(p)$ such that $\mathit{target}(A) \in \mathrm{match}(\mathit{pat})$ and $\mathit{type}(A) = \mathit{type}$, then $A$ is rejected.
\end{axiom}

\begin{definition}[Root Privilege]
The root actor $r$ has initial writability $W(r) = \{(\texttt{**}, \texttt{*})\}$ (full wildcard). Privileged targets (\texttt{system/*}, \texttt{ledger/*}) are writable only by root.
\end{definition}

\begin{definition}[Envelope]
An envelope is a temporary authorization structure $\varepsilon = (\mathit{issuer}, \mathit{budget}, \mathit{targets}, \mathit{actions}, \mathit{duration}, \mathit{checkpoint}, \mathit{hold\_on}, \mathit{hold\_timeout})$, satisfying $\mathit{targets}(\varepsilon) \times \mathit{actions}(\varepsilon) \subseteq W(\mathit{issuer}(\varepsilon))$ (the permission subset constraint).
\end{definition}

\begin{remark}
Only human actors ($\actorsh$) may create agents (enforced by the kernel; agents attempting to create agents receive a policy violation error). However, an agent holding an active envelope \emph{may} issue a sub-envelope to another existing agent, subject to \emph{envelope reduction validation}: the sub-envelope's targets, actions, and budget must be a subset of the issuing envelope's remaining authorization. This yields a two-tier model: agent \emph{creation} is human-only, while envelope \emph{delegation} supports constrained chains with monotonically decreasing privileges.
\end{remark}

\subsection{Energy Model}
\label{sec:energy}

The energy model serves two concrete purposes: (1)~it provides a denial-of-service defense by limiting each actor's action rate to a hardware-anchored budget, preventing a single agent from monopolizing the commit pipeline; and (2)~it creates an economic constraint on the hold mechanism, ensuring that repeated hold requests have a real cost that deters boundary probing (threat T2). Energy is not a currency to be traded; it is a resource accounting mechanism analogous to process CPU quotas in operating systems, but operating at the semantic action level rather than the instruction level.

\begin{definition}[Hardware Capacity]
At startup, the kernel reads the hardware compute capacity $\lambda$ (measured in INT8 Tera Operations Per Second (TOPS) or a user-configured value) from a configuration file, serving as the physical anchor for energy production.
\end{definition}

\begin{definition}[Continuous Production]
At each tick $t$, the system produces energy $\mathrm{produce}(t) = \lambda \cdot \Delta t$, where $\Delta t$ is the tick duration. The production event is recorded in the log.
\end{definition}

\begin{definition}[Share-Based Allocation]
Each actor $p$ holds a share $\mathrm{share}(p) \in \mathbb{R}_{>0}$. Per-tick energy received:
\[
\mathrm{receive}(p, t) = \frac{\mathrm{share}(p)}{\sum_{q \in \actors} \mathrm{share}(q)} \cdot \mathrm{produce}(t)
\]
\end{definition}

\begin{definition}[Cost Function]
The cost function $\mathrm{cost}: \mathit{ActionType} \to \mathbb{R}_{\geq 0}$ is a deployment-configurable, monotonically ordered function satisfying: $\mathrm{cost}(\texttt{observe}) \leq \mathrm{cost}(\texttt{create}) \leq \mathrm{cost}(\texttt{mutate}) \leq \mathrm{cost}(\texttt{execute})$. The \texttt{execute} cost may include an I/O-proportional component based on the \texttt{output\_bytes} field reported by the submitting actor.
\end{definition}

\begin{remark}
The reference implementation uses $\mathrm{cost}(\texttt{observe}) = 0$, $\mathrm{cost}(\texttt{create}) = 10$, $\mathrm{cost}(\texttt{mutate}) = 15$, $\mathrm{cost}(\texttt{execute}) = 25 + \lfloor\mathit{output\_bytes}/256\rfloor$. These values are derived from profiling the relative pipeline overhead of each action type; the ordering reflects the principle that state-changing actions should cost more than read-only ones. The formal framework is parametric in these values: all invariants hold for any assignment satisfying the monotonicity constraint.
\end{remark}

\begin{axiom}[Production Sufficiency]
$\lambda \cdot \Delta t \geq \max_{\mathit{type}}\, \mathrm{cost}(\mathit{type})$, ensuring the Right to History is not undermined by energy design.
\end{axiom}

\begin{definition}[Envelope Budget Constraint]
For an active envelope $\varepsilon$: $\mathrm{consumed}(\varepsilon) + \mathrm{reserved}(\varepsilon) \leq \mathrm{budget}(\varepsilon)$, where $\mathrm{consumed}$ is settled energy and $\mathrm{reserved}$ is energy locked by pending holds.
\end{definition}

\subsection{Hold Mechanism}
\label{sec:hold}

\begin{definition}[Hold Rules]
An envelope $\varepsilon$ may declare hold rules $H(\varepsilon) \subseteq \mathit{Pattern} \times \mathit{ActionType}$.
\end{definition}

\begin{definition}[Hold Trigger]
When action $A$ executes under envelope $\varepsilon$ and $(\mathit{target}(A), \mathit{type}(A))$ matches $H(\varepsilon)$, the kernel atomically: (1)~computes $\mathrm{cost}(A)$, (2)~reserves $\mathrm{cost}(A)$ energy, (3)~appends a \texttt{hold\_request} event to the log, and (4)~returns a \texttt{HoldTriggered} status. The envelope remains \emph{Active}---the agent may continue submitting non-held actions.
\end{definition}

\begin{definition}[Hold Response]
A human actor responds via \texttt{mutate ledger/hold/\textlangle id\textrangle}:
\begin{itemize}[nosep]
\item \textbf{Approve}: execute the original action $\to$ settle $\to$ append result event.
\item \textbf{Reject}: settle commitment cost ($20\%$ of reserved) $\to$ release remainder $\to$ append \texttt{hold\_response} event.
\item \textbf{Timeout} (if \texttt{hold\_timeout\_secs} set): equivalent to reject.
\end{itemize}
\end{definition}

\begin{definition}[Commitment Cost]
On reject or timeout: $\mathrm{commitment\_cost} = \lceil 0.2 \times \mathrm{reserved\_cost} \rceil$. This reflects hardware resources consumed by validation, quoting, reservation, and recording.
\end{definition}

\begin{remark}
The commitment cost reflects a real resource expenditure: even a rejected or timed-out hold has already consumed kernel resources---validation, quoting, energy reservation, and writing a \texttt{hold\_request} event to the append-only log (analogous to reverted transactions still consuming gas in Ethereum). The 20\% rate is a configurable engineering parameter balancing two extremes: 0\% allows zero-cost boundary probing (enabling T2 budget exhaustion); 100\% penalizes legitimate ``request-then-approve'' workflows. The current value is based on profiling typical agent workloads, where the pre-decision pipeline stages consume approximately 15--25\% of total action cost. A formal game-theoretic analysis of the optimal commitment rate under adversarial agent strategies is left as future work.
\end{remark}

\subsection{System Invariants and Proofs}
\label{sec:invariants}

This section states and argues five core invariants for the \PunkGo{} system. Arguments follow a structured pre-/post-condition style based on the state transition semantics defined in \S\ref{sec:world}--\S\ref{sec:hold}. These are proof sketches rather than machine-checked proofs; full verification using Verus~\cite{verus} is future work (\S\ref{sec:discussion}).

\begin{invariant}[Append-Only]
\label{inv:append}
$\forall\, t_1 < t_2$: $\logL(t_1) \sqsubseteq \logL(t_2)$, i.e., $\logL(t_1)$ is a prefix of $\logL(t_2)$: $\forall\, i \leq |\logL(t_1)|: \logL(t_2)[i] = \logL(t_1)[i]$.
\end{invariant}

\begin{proof}[Proof sketch]
\emph{(Step 1)} By Definition~\ref{def:committer}, the kernel is the sole committer. The only operation that modifies $\logL$ is \textsc{Append}, whose postcondition is $\logL' = \logL \| \langle e \rangle$---concatenation at the end, without modifying existing elements.

\emph{(Step 2)} No \textsc{Delete} or \textsc{Modify} operation exists in the system. By Axiom~\ref{ax:deny} (default deny), \texttt{ledger/*} targets are writable only by root, and root's operations also pass through the kernel pipeline, where the only available mutation is \textsc{Append}.

\emph{(Step 3)} The Merkle tree structure provides additional protection: modifying event $e_i$ changes its leaf hash $H(0\mathtt{x00} \| \mathrm{canonical}(e_i))$, invalidating every intermediate node from that leaf to the root. Under the collision-resistance assumption (A4), an attacker cannot find alternative content that preserves the leaf hash.

Therefore $\logL$ grows monotonically.
\end{proof}

\begin{invariant}[Completeness]
\label{inv:complete}
$\forall$ action $A$: if $A$ passes boundary check (i.e., $(\mathit{target}(A), \mathit{type}(A)) \in W(\mathit{actor}(A))$) and has sufficient energy (i.e., passes \texttt{reserve}), then $\exists\, e \in \logL$ recording $A$'s commit or hold status.
\end{invariant}

\begin{proof}[Proof sketch]
\emph{(Step 1)} The kernel's action pipeline is a deterministic linear sequence:
\[
\texttt{validate} \to \texttt{quote} \to \texttt{reserve} \to \texttt{validate\_payload} \to \texttt{settle} \to \texttt{append} \to \texttt{receipt}
\]

\emph{(Step 2)} Passing boundary check means $A$ passes \texttt{validate}. Two paths follow: (a)~\emph{normal path}: \texttt{quote} $\to$ \texttt{reserve} $\to$ \texttt{validate\_payload} $\to$ \texttt{settle} $\to$ \textbf{append}; (b)~\emph{hold path}: validate internally performs quote $+$ reserve $+$ \textbf{append hold\_request} $\to$ \texttt{HoldTriggered} $\to$ \{approve $\to$ re-submit with replay guards $\to$ full pipeline $\to$ \textbf{append}\} $|$ \{reject/timeout $\to$ settle commitment $\to$ \textbf{append hold\_response}\}.

\emph{(Step 3)} Both the normal path and the hold path terminate with some form of append. Actions that pass boundary check but lack sufficient energy are rejected at \texttt{reserve} with \texttt{InsufficientEnergy}; these are not recorded in the permanent log, as Energy Conservation (INV-5) takes precedence---the system prioritizes resource integrity over recording resource-denied actions. Completeness thus covers all actions that enter the funded pipeline.

Therefore every boundary-checked, energy-funded action is recorded.
\end{proof}

\begin{invariant}[Integrity]
\label{inv:integrity}
$\forall\, t$: $\mathrm{root}(t) = \mathrm{MerkleRoot}(\logL(t))$, i.e., the Merkle root precisely commits the current log contents.
\end{invariant}

\begin{proof}[Proof sketch]
\emph{(Base case)} $\logL = \langle\rangle$, $\mathrm{root} = \mathit{null}$. $\mathrm{MerkleRoot}(\langle\rangle) = \mathit{null}$. Holds.

\emph{(Inductive step)} Assume $\mathrm{root}(t) = \mathrm{MerkleRoot}(\logL(t))$. When \textsc{Append} occurs with event $e$: $\logL(t{+}1) = \logL(t) \| \langle e \rangle$ and $\mathrm{root}(t{+}1) = \mathrm{MerkleUpdate}(\mathrm{root}(t), e)$. By the correctness of the RFC~6962 Merkle tree construction, $\mathrm{MerkleUpdate}$ computes precisely $\mathrm{MerkleRoot}(\logL(t) \| \langle e \rangle)$.

Root is updated only during \textsc{Append}. By INV-1, no other operation modifies $\logL$ or root. Therefore root always reflects $\logL$.
\end{proof}

\begin{corollary}[Anti-Equivocation]
If two verifiers obtain $\mathrm{root}_1 \neq \mathrm{root}_2$ from the kernel, at least one does not correspond to the kernel's true log---equivocation is detected. Specifically, for a claimed prefix relation $\logL \sqsubseteq \logL'$, a consistency proof (Definition~4.10) can verify that $\logL$ is indeed a prefix of $\logL'$; verification failure proves tampering or forking. \emph{Caveat:} this property requires at least two independent verifiers comparing roots. In the single-user, single-machine deployment (the current \PunkGo{} scenario), external root anchoring---e.g., periodically publishing roots to a public service or a trusted third party---is needed to realize anti-equivocation in practice. The kernel's Audit Prover already exports signed tree heads in the C2SP checkpoint format, providing the architectural hook for such anchoring.
\end{corollary}

\begin{invariant}[Boundary Enforcement]
\label{inv:boundary}
$\forall$ action $A$ by actor $p$: $(\mathit{target}(A), \mathit{type}(A)) \notin W(p) \implies A$ is not committed and produces no log entry.
\end{invariant}

\begin{proof}[Proof sketch]
\emph{(Step 1)} \texttt{validate} is the first pipeline step. Its postcondition: if $(\mathit{target}(A), \mathit{type}(A)) \in W(\mathit{actor}(A))$, return \textsc{Validated}; otherwise return \textsc{Rejected} and terminate the pipeline.

\emph{(Step 2)} The \texttt{settle} and \texttt{append} steps (which commit the action to the permanent log) are reachable only after \texttt{validate} returns \textsc{Validated}.

\emph{(Step 3)} The \textsc{Rejected} path does not pass through \texttt{settle} or \texttt{append}. The kernel does not record a rejection event in the append-only log (the action never entered the committed pipeline), ensuring no state changes occur.

\emph{Design note:} Not recording rejections is a deliberate trade-off. Security audit systems (e.g., auditd) typically log failed attempts to detect probing attacks. In \PunkGo{}, the energy system provides an alternative defense: each attempt (including those targeting holds) consumes energy, making sustained probing self-limiting. Optionally logging rejections as zero-cost observe events is a configurable extension that does not affect INV-4's guarantee.

Therefore out-of-bounds actions cannot be committed.
\end{proof}

\begin{invariant}[Energy Conservation]
\label{inv:energy}
$\forall$ actor $p$, $\forall\, t$: $E(p, t) \geq 0$.
\end{invariant}

\begin{proof}[Proof sketch]
\emph{(Step 1)} \texttt{quote} computes $\mathrm{cost}(A)$. \texttt{reserve} checks: $E(\mathit{actor}(A)) - \mathrm{reserved}(\mathit{actor}(A)) \geq \mathrm{cost}(A)$. If unsatisfied, \texttt{InsufficientEnergy} is returned and $A$ is not committed.

\emph{(Step 2)} On successful reservation: $\mathrm{reserved} \mathrel{+}= \mathrm{cost}(A)$. The sole-committer guarantee (Definition~\ref{def:committer}) ensures no concurrent operation can double-spend the same energy.

\emph{(Step 3)} \texttt{settle} converts reserved to consumed: $E(p) \mathrel{-}= \mathrm{actual\_cost}$. For hold reject/timeout: $E(p) \mathrel{-}= \lceil 0.2 \times \mathrm{reserved\_cost}\rceil$. In both cases, the deduction does not exceed the amount verified available at reservation time.

Therefore $E(p, t) \geq 0$ always holds.
\end{proof}

\paragraph{Invariant Dependencies.}
The five invariants form a security chain (Figure~\ref{fig:chain}): Boundary Enforcement (INV-4) ensures only legitimate actions enter the pipeline; Energy Conservation (INV-5) ensures legitimate actions have sufficient resources; Completeness (INV-2) ensures every executed or held action is recorded; Append-Only (INV-1) ensures records cannot be tampered with; Integrity (INV-3) ensures tamper-evidence is independently verifiable. Compromising any invariant requires first compromising its upstream invariant.

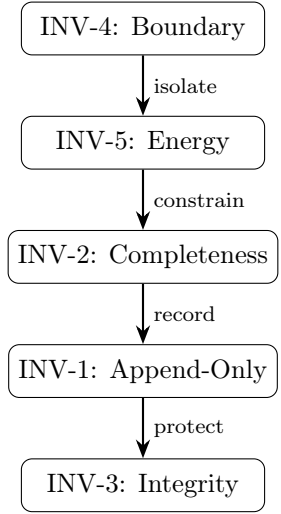
\begin{figure}[t]
\centering
\begin{tikzpicture}[
  node distance=0.8cm,
  box/.style={draw, rounded corners, minimum width=3.2cm, minimum height=0.7cm, align=center, font=\small},
  arr/.style={-{Stealth[length=2.5mm]}, thick}
]
\node[box] (inv4) {INV-4: Boundary};
\node[box, below=of inv4] (inv5) {INV-5: Energy};
\node[box, below=of inv5] (inv2) {INV-2: Completeness};
\node[box, below=of inv2] (inv1) {INV-1: Append-Only};
\node[box, below=of inv1] (inv3) {INV-3: Integrity};

\draw[arr] (inv4) -- node[right, font=\scriptsize] {isolate} (inv5);
\draw[arr] (inv5) -- node[right, font=\scriptsize] {constrain} (inv2);
\draw[arr] (inv2) -- node[right, font=\scriptsize] {record} (inv1);
\draw[arr] (inv1) -- node[right, font=\scriptsize] {protect} (inv3);
\end{tikzpicture}
\caption{The security chain formed by the five invariants. Each invariant depends on the one above it.}
\label{fig:chain}
\end{figure}

\section{System Design: PunkGo Kernel}
\label{sec:design}

\subsection{Architecture Overview}

Figure~\ref{fig:arch} illustrates \PunkGo{}'s three-layer architecture.

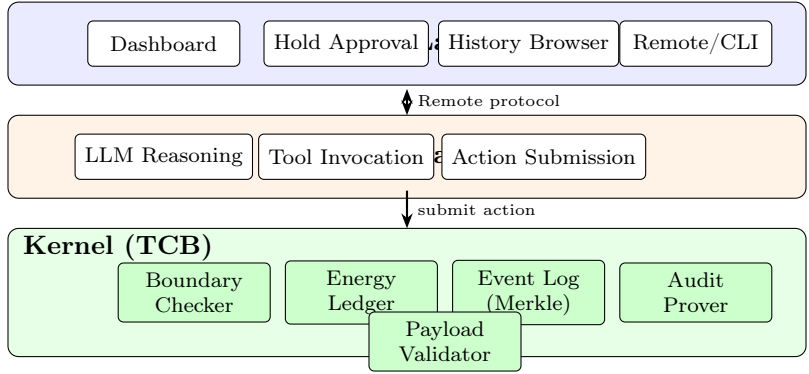
\begin{figure}[t]
\centering
\begin{tikzpicture}[
  layer/.style={draw, rounded corners=4pt, minimum width=10.5cm, minimum height=1.1cm, align=center, font=\small},
  comp/.style={draw, rounded corners=2pt, minimum width=2.0cm, minimum height=0.6cm, align=center, font=\scriptsize, fill=white},
  arr/.style={-{Stealth[length=2mm]}, thick},
  darr/.style={{Stealth[length=2mm]}-{Stealth[length=2mm]}, thick}
]
\node[layer, fill=blue!8] (client) at (0,3.6) {\textbf{Client Layer}};
\node[comp] at (-3.2,3.6) {Dashboard};
\node[comp] at (-0.8,3.6) {Hold Approval};
\node[comp] at (1.6,3.6) {History Browser};
\node[comp] at (3.8,3.6) {Remote/CLI};

\node[layer, fill=orange!10] (agent) at (0,2.1) {\textbf{Agent Layer}};
\node[comp] at (-3.2,2.1) {LLM Reasoning};
\node[comp] at (-0.8,2.1) {Tool Invocation};
\node[comp] at (1.8,2.1) {Action Submission};

\node[layer, fill=green!10, minimum height=1.7cm] (kernel) at (0,0.3) {};
\node[font=\small\bfseries] at (-3.8,0.9) {Kernel (TCB)};
\node[comp, fill=green!20] at (-2.8,0.3) {Boundary\\Checker};
\node[comp, fill=green!20] at (-0.6,0.3) {Energy\\Ledger};
\node[comp, fill=green!20] at (1.6,0.3) {Event Log\\(Merkle)};
\node[comp, fill=green!20] at (3.8,0.3) {Audit\\Prover};
\node[comp, fill=green!20] at (0.5,-0.35) {Payload\\Validator};

\draw[darr] (0,3.0) -- (0,2.65) node[right, font=\tiny, pos=0.5] {Remote protocol};
\draw[arr] (0,1.65) -- (0,1.15) node[right, font=\tiny, pos=0.5] {submit action};
\end{tikzpicture}
\caption{Three-layer architecture of \PunkGo{}. The kernel layer (green, TCB) is the sole commit point; agents submit actions but cannot directly access system resources.}
\label{fig:arch}
\end{figure}

\PunkGo{} follows a three-layer architecture:

\paragraph{Kernel Layer (TCB).} A Rust user-space process serving as the sole commit point for all AI agent actions. Components: Event Log (Merkle tree), Energy Ledger, Boundary Checker, Payload Validator, and Audit Prover. The kernel is a \emph{committer, not an executor}: it validates, records, and proves actions, but does not spawn or run external processes on behalf of any actor.

\paragraph{Agent Layer.} Receives commands from the client layer (locally via CLI, or remotely via a messaging protocol such as Matrix) and submits actions to the kernel below. Agents handle LLM reasoning, tool invocation, and execution; they submit action results to the kernel for recording. Agents do not directly manipulate kernel state---all committed records must pass through the kernel pipeline.

\paragraph{Client Layer.} Remote control entry point communicating with agents via a messaging protocol (e.g., Matrix). Provides dashboard, hold-approval interface, and history browser. The client does not interact directly with the kernel.

A key design principle: \emph{the kernel is command-source agnostic}. Whether a command arrives from a local CLI or a remote messaging protocol does not affect the kernel's validation, execution, and recording logic, ensuring audit consistency.

\subsection{Action Pipeline}

The kernel processes each action through its \texttt{submit\_action} entry point, a 7-step pipeline:

\begin{enumerate}[nosep]
\item \textbf{Validate}: Check actor existence and active status; verify that the action target falls within the actor's writable boundary $((\mathit{target}, \mathit{type}) \in W(\mathit{actor}))$; for state-changing actions, check envelope authorization. [INV-4]
\item \textbf{Quote}: Compute $\mathrm{cost}(A)$ via \texttt{quote\_cost}.
\item \textbf{Reserve}: Check envelope available energy $\geq \mathrm{cost}(A)$ and lock the required amount. [INV-5]
\item \textbf{Validate Payload}: For \texttt{execute} actions, validate that all required fields are present with correct types and that OID format is valid (\texttt{sha256:} followed by 64 hex characters). For other action types, validate the corresponding payload constraints.
\item \textbf{Settle}: Convert locked energy to consumed, finalizing energy accounting.
\item \textbf{Append}: Write the committed event to the append-only log, compute the RFC~6962 leaf hash, update the Merkle tree, and write the audit leaf. [INV-1, INV-3]
\item \textbf{Receipt}: Return a verifiable receipt: \texttt{event\_id} (UUID), \texttt{log\_index} (sequence number), \texttt{event\_hash} (SHA-256).
\end{enumerate}

The pipeline is synchronous, linear, and deterministic---a direct consequence of the sole-committer design. Any step failure terminates the pipeline: \texttt{validate} failure returns \texttt{Rejected}; \texttt{reserve} failure returns \texttt{InsufficientEnergy}. Failed paths produce no log entries.

\paragraph{Hold path.} Hold is not a branch \emph{within} the pipeline; it is a \textbf{short-circuit exit inside validate}. When \texttt{validate} detects that an action matches the envelope's \texttt{hold\_on} rules, the kernel performs the following operations \emph{inside} the validate step, within a single database transaction:

\begin{enumerate}[nosep]
\item Call \texttt{quote\_cost} to compute the energy cost,
\item Reserve energy on the envelope,
\item Append a \texttt{hold\_request} event to the log and audit tree,
\item Create a hold request record,
\item Commit the transaction and return \texttt{HoldTriggered}.
\end{enumerate}

Normal pipeline steps 2--7 are unreachable for held actions. From an abstract perspective, the hold path is: $\texttt{validate}(\text{including quote} + \text{reserve} + \text{append hold\_request}) \to \texttt{HoldTriggered}$.

\paragraph{Hold response.} A human responds via \texttt{mutate ledger/hold/<id>}:

\emph{Approve.} The kernel recursively re-submits the original action to \texttt{submit\_action} with two payload markers (replay guards): (1)~\texttt{\_hold\_approved}, which causes the hold-rule check inside validate to become a no-op, and (2)~\texttt{\_hold\_reserved\_cost}, which causes quote and reserve to produce a phantom reservation matching the amount locked at hold time, preventing double-charging. The re-submitted action traverses the full pipeline: validate re-checks actor status, boundary permissions, and envelope authorization (which may have changed during the hold period), while steps 4--7 execute normally. This design avoids a special ``resume from step~4'' path; instead, it reuses the same \texttt{submit\_action} function with replay guards that make already-completed steps idempotent.

\emph{Reject.} Settle commitment cost ($\lceil 0.2 \times \mathit{reserved\_cost} \rceil$), release remaining energy, and append a \texttt{hold\_response} event.

\emph{Timeout} (if \texttt{hold\_timeout\_secs} is configured): equivalent to reject.

Figure~\ref{fig:pipeline} illustrates the flow.

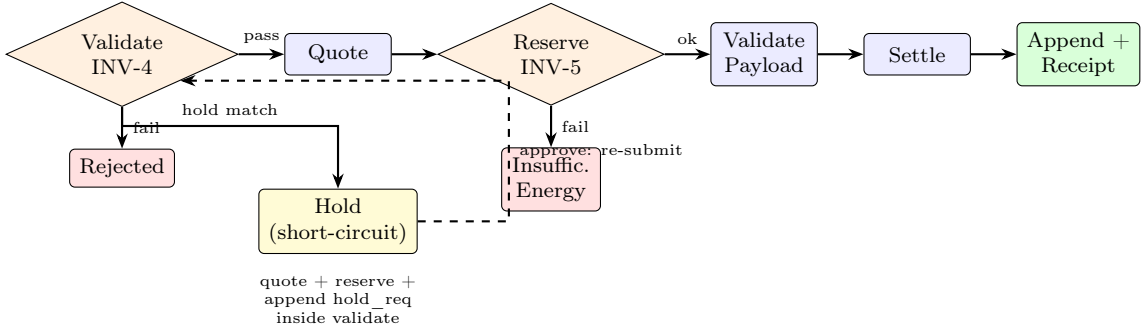
\begin{figure}[t]
\centering
\begin{tikzpicture}[
  stp/.style={draw, rounded corners=2pt, minimum width=1.4cm, minimum height=0.55cm, align=center, font=\scriptsize, fill=blue!8},
  decision/.style={draw, diamond, aspect=2.2, minimum width=1.2cm, align=center, font=\scriptsize, fill=orange!12},
  result/.style={draw, rounded corners=2pt, minimum width=1.3cm, minimum height=0.5cm, align=center, font=\scriptsize, fill=green!15},
  fail/.style={draw, rounded corners=2pt, minimum width=1.3cm, minimum height=0.5cm, align=center, font=\scriptsize, fill=red!12},
  hold/.style={draw, rounded corners=2pt, minimum width=1.3cm, minimum height=0.5cm, align=center, font=\scriptsize, fill=yellow!20},
  arr/.style={-{Stealth[length=2mm]}, thick},
  darr/.style={-{Stealth[length=2mm]}, thick, dashed},
  node distance=0.45cm and 0.6cm
]
\node[decision] (validate) {Validate\\INV-4};
\node[stp, right=of validate] (quote) {Quote};
\node[decision, right=of quote] (reserve) {Reserve\\INV-5};
\node[stp, right=of reserve] (valpay) {Validate\\Payload};
\node[stp, right=of valpay] (settle) {Settle};
\node[result, right=of settle] (append) {Append +\\Receipt};

\node[fail, below=0.55cm of validate] (reject) {Rejected};
\node[fail, below=0.55cm of reserve] (insuf) {Insuffic.\\Energy};
\node[hold, below=1.5cm of quote] (held) {Hold\\(short-circuit)};
\node[font=\tiny, below=0.15cm of held, text width=2.5cm, align=center] (heldnote) {quote + reserve +\\append hold\_req\\inside validate};

\draw[arr] (validate) -- node[above, font=\tiny] {pass} (quote);
\draw[arr] (validate) -- node[right, font=\tiny] {fail} (reject);
\draw[arr] (validate.south) -- ++(0,-0.25) -| node[above, font=\tiny, pos=0.25] {hold match} (held);
\draw[arr] (quote) -- (reserve);
\draw[arr] (reserve) -- node[above, font=\tiny] {ok} (valpay);
\draw[arr] (reserve) -- node[right, font=\tiny] {fail} (insuf);
\draw[arr] (valpay) -- (settle);
\draw[arr] (settle) -- (append);
\draw[darr] (held.east) -- ++(1.2,0) |- node[right, font=\tiny, pos=0.25] {approve: re-submit} (validate.south east);
\end{tikzpicture}
\caption{Action pipeline. The normal path traverses all seven steps left to right. The hold path short-circuits inside validate (which internally performs quote, reserve, and append of the hold request), returning \texttt{HoldTriggered}. On approve, the action is recursively re-submitted to the full pipeline with replay guards that make the hold check and energy reservation idempotent.}
\label{fig:pipeline}
\end{figure}

\subsection{Append-Only Event Log}

Persistence is backed by a configurable append-only storage backend (the reference implementation uses SQLite in WAL mode, but the design does not bind to any specific database---any store providing atomic writes and durability guarantees is substitutable). The Merkle tree implementation is based on the Google tlog algorithm (C2SP tlog-checkpoint signed tree head format):

\begin{itemize}[nosep]
\item Leaf nodes: $H(0\mathtt{x00} \| \mathit{event\_data})$ (RFC~6962 \S2.1)
\item Internal nodes: $H(0\mathtt{x01} \| \mathit{left} \| \mathit{right})$
\end{itemize}

Two proof types are supported: \emph{inclusion proofs} ($O(\log n)$, verifying a specific event's existence) and \emph{consistency proofs} ($O(\log n)$, verifying the old log is a prefix of the new log). Users can export any time range of events plus Merkle proofs as a standalone, independently verifiable audit package.

\subsection{Energy Ledger}

The energy system anchors hardware compute capacity as a virtual resource:

\begin{itemize}[nosep]
\item \textbf{Hardware capacity}: Read from a configuration file at startup; runtime hardware probing is future work.
\item \textbf{Tick production}: Energy produced at fixed intervals, recorded in the log.
\item \textbf{Share allocation}: Distributed proportionally among actors.
\item \textbf{Envelope budgets}: Creators allocate budgets to agents, with permissions that cannot exceed the creator's own.
\end{itemize}

The 20\% commitment cost on rejected or timed-out holds reflects actual resource consumption and deters boundary probing (see \S\ref{sec:hold} for the formal rationale).

\subsection{Hold/Approval Workflow}

Table~\ref{tab:hold-comparison} contrasts \PunkGo{}'s hold mechanism with existing approaches.

\begin{table}[t]
\centering
\caption{Comparison of human-approval mechanisms.}
\label{tab:hold-comparison}
\small
\begin{tabularx}{\columnwidth}{lXcc}
\toprule
\textbf{System} & \textbf{Mechanism} & \textbf{Econ.} & \textbf{Audit} \\
\midrule
Terraform plan & Full-batch pause & No & No \\
LangGraph & \texttt{interrupt()} & No & No \\
OpenAI SDK & Boolean flag & No & No \\
\PunkGo{} & Rule-matched hold & 20\% & Merkle \\
\bottomrule
\end{tabularx}
\end{table}

\PunkGo{}'s hold is not merely flow control; it is an economic constraint that prevents agents from probing system boundaries via mass hold requests.

\subsection{Execute Submission}

The kernel is a \emph{committer, not an executor}. For \texttt{execute} actions, the actor performs the execution externally and submits the result to the kernel. The kernel validates the payload format and records the event---it does not spawn, run, or wait for external processes on behalf of any actor.

The execute payload must contain:

\begin{itemize}[nosep]
\item \texttt{input\_oid}: \texttt{sha256:<64 hex>} --- reference to serialized input (command, args, env, working directory)
\item \texttt{output\_oid}: \texttt{sha256:<64 hex>} --- reference to serialized output (stdout, stderr)
\item \texttt{exit\_code}: integer --- process exit code
\item \texttt{artifact\_hash}: \texttt{sha256:<64 hex>} --- SHA-256 of the raw output bytes
\item \texttt{output\_bytes}: integer (optional) --- output size for IO cost calculation
\end{itemize}

The kernel validates that all required fields are present with correct types and that OID format is valid, but does \emph{not} verify whether OID references point to actual content. The kernel records what the actor claims, consistent with \texttt{mutate} behavior. The implications of this trust boundary---and the analogy to Certificate Transparency's approach---are discussed in \S\ref{sec:discussion}. Energy cost: $25 + \lfloor\mathit{output\_bytes}/256\rfloor$.

This design maintains symmetry across all four action types: for \texttt{mutate}, the actor writes files and submits metadata (\texttt{content\_oid}); for \texttt{execute}, the actor performs the execution and submits the result. In both cases, the actor acts, the kernel records. Execution environment management (sandboxing, working directory isolation, timeout enforcement) is the actor's responsibility.

\section{Evaluation}
\label{sec:eval}

\PunkGo{} is evaluated along two dimensions: qualitative invariant verification and quantitative performance characterization. All benchmarks use fresh kernel instances (clean databases) and are compiled with \texttt{--release} optimization.

\subsection{Experimental Setup}

\begin{itemize}[nosep]
\item \textbf{Hardware}: Intel Core i7-14700K (14th gen, 20 cores / 28 threads, 32 logical), 64\,GB DDR5-5600 RAM, NVMe SSD.
\item \textbf{Software}: Windows~11 Pro (build 26100), \PunkGo{} Kernel v0.2.1, Rust~1.90.0 (release mode), SQLite WAL backend.
\item \textbf{Code}: Open-sourced at \url{https://github.com/PunkGo/punkgo-kernel}. Benchmark scripts included in the repository.
\item \textbf{Methodology}: Latency tests use 200 iterations (20 warmup discarded), reporting medians and 95th percentiles. Throughput tests use 5-second windows. All benchmarks are single-threaded (sole-committer design).
\end{itemize}

\subsection{Invariant Verification}

Adversarial test scenarios were designed for each invariant. All five pass:

\paragraph{INV-1 (Append-Only).} The test submitted 5 actions, then bypassed the kernel to directly modify the third event's payload in the underlying store. The recomputed event hash diverged from the stored Merkle leaf hash (\texttt{f79d...} $\to$ \texttt{6a3e...}). Inclusion proof verification against the original root: \textbf{failed} (expected). Tamper detection confirmed in 55\,ms.

\paragraph{INV-2 (Completeness).} The test submitted 20 actions: 5$\times$observe, 5$\times$create, 5$\times$mutate, and 5$\times$out-of-bounds mutate (targeting \texttt{system/*}). All 15 legitimate actions produced event records; 5 violations were rejected at \texttt{validate} (pipeline entry), producing no log entries---correct behavior, as INV-2 covers only boundary-checked actions. Hold mechanism tested separately: both \texttt{hold\_request} and \texttt{hold\_response} events recorded. 100\,ms.

\paragraph{INV-3 (Integrity).} The test submitted 10 actions and generated inclusion proofs for each, verifying with an \emph{independent} RFC~6962 verifier (not kernel code). All 10/10 passed. Proof lengths matched expectations: $\lceil \log_2(10) \rceil = 4$ hashes (average 3.6, as trailing leaves have shorter paths). After appending 5 more events, consistency proof (old=10, new=15) verified successfully. 72\,ms.

\paragraph{INV-4 (Boundary Enforcement).} The test created an agent with \texttt{writable\_targets} restricted to \texttt{workspace/docs/*} (mutate only) and tested 7 scenarios covering: in-bounds access, out-of-bounds target, disallowed action type, privileged targets (\texttt{system/*}, \texttt{ledger/*}), observe exemption, and root override. All 7/7 matched expected behavior. 55\,ms.

\paragraph{INV-5 (Energy Conservation).} Envelope with budget=100, cost-per-mutate=15. Six successful mutates (balance: 100$\to$85$\to$70$\to$55$\to$40$\to$25$\to$10); seventh rejected with \texttt{InsufficientEnergy}, balance unchanged at 10. Hold commitment test: budget=1000, hold reserved=15, reject settled $\lceil 15 \times 0.2 \rceil = 3$, balance $1000 \to 997$. Balance never negative. 96\,ms.

\subsection{Performance Characteristics}

\paragraph{Action Pipeline Latency.}
End-to-end latency encompasses the full pipeline (validate through receipt, plus post-commit audit checkpoint). Table~\ref{tab:latency} reports results over 200~iterations.

\begin{table}[t]
\centering
\caption{Action pipeline latency ($n=200$, 20 warmup).}
\label{tab:latency}
\small
\begin{tabular}{lrrrr}
\toprule
\textbf{Type} & \textbf{Median} & \textbf{P95} & \textbf{Min} & \textbf{Max} \\
\midrule
\texttt{observe} & 665\,$\mu$s & 876\,$\mu$s & 463\,$\mu$s & 1{,}759\,$\mu$s \\
\texttt{create}  & 1{,}265\,$\mu$s & 1{,}670\,$\mu$s & 990\,$\mu$s & 3{,}310\,$\mu$s \\
\texttt{mutate}  & 1{,}274\,$\mu$s & 1{,}681\,$\mu$s & 994\,$\mu$s & 15{,}470\,$\mu$s \\
\texttt{execute} & 1{,}258\,$\mu$s & 1{,}807\,$\mu$s & 1{,}003\,$\mu$s & 2{,}725\,$\mu$s \\
\bottomrule
\end{tabular}
\end{table}

Median latency is 0.7--1.3\,ms per action, dominated by storage I/O (event persistence and Merkle tree update). The \texttt{execute} action type adds a \texttt{validate\_payload} step (OID format checking) but exhibits latency comparable to \texttt{mutate}, confirming that payload validation adds negligible overhead. Medians and percentiles are reported rather than means with error bars because latency distributions are right-skewed; percentiles better characterize tail behavior.

\paragraph{Merkle Proof Scaling.}
Table~\ref{tab:merkle} shows inclusion proof generation time and size as the log grows.

\begin{table}[t]
\centering
\caption{Merkle inclusion proof scaling.}
\label{tab:merkle}
\small
\begin{tabular}{rrrr}
\toprule
\textbf{Log Size} & \textbf{Median Time ($\mu$s)} & \textbf{Hashes} & \textbf{Bytes} \\
\midrule
10     & 30    & 4  & 128 \\
50     & 87    & 6  & 192 \\
100    & 158   & 7  & 224 \\
500    & 734   & 9  & 288 \\
1{,}000 & 1{,}445 & 10 & 320 \\
10{,}000 & 19{,}332 & 14 & 448 \\
\bottomrule
\end{tabular}
\end{table}

Proof size (hash count) follows $O(\log n)$ precisely: from 4 hashes at 10 entries to 14 hashes at 10{,}000 ($\lceil \log_2 10{,}000 \rceil = 14$). Each hash is 32 bytes (SHA-256), so a 10{,}000-entry proof is only 448~bytes. Generation time grows as $O(n)$ because the current implementation loads all hash nodes from storage before computing the $O(\log n)$ proof path---an engineering optimization opportunity (paginated loading would restore $O(\log n)$ time). At 10{,}000 entries, proof generation takes $\sim$19\,ms; 100{,}000-entry benchmarks were attempted but the $O(n^2)$ cumulative fill cost (each commit triggers $O(n)$ hash loading) made the test impractical, confirming that the hash-loading optimization is the critical scaling bottleneck.

\paragraph{Throughput.}
Table~\ref{tab:throughput} reports sequential single-committer throughput over 5-second windows.

\begin{table}[t]
\centering
\caption{Throughput (5-second window, sequential commits).}
\label{tab:throughput}
\small
\begin{tabular}{lrr}
\toprule
\textbf{Scenario} & \textbf{Actions/sec} & \textbf{5\,s Total} \\
\midrule
\texttt{observe} only & 450 & 2{,}248 \\
\texttt{create} only  & 396 & 1{,}981 \\
\texttt{mutate} only  & 392 & 1{,}958 \\
\texttt{execute} only & 393 & 1{,}968 \\
Mixed                  & 412 & 2{,}059 \\
\bottomrule
\end{tabular}
\end{table}

Throughput is $\sim$390--450 actions/sec. The bottleneck is storage I/O (event persistence and Merkle checkpoint). For personal AI agent usage (typical daily volume in the hundreds to low thousands), current throughput is sufficient.

\paragraph{Hold Workflow Latency.}
Table~\ref{tab:hold} reports mechanical latency only, excluding human decision time.

\begin{table}[t]
\centering
\caption{Hold workflow mechanical latency ($n=8$).}
\label{tab:hold}
\small
\begin{tabular}{lrr}
\toprule
\textbf{Phase} & \textbf{Median ($\mu$s)} & \textbf{Max ($\mu$s)} \\
\midrule
Hold trigger     & 751   & 1{,}323 \\
Read pending     & 22    & 518 \\
Approve          & 2{,}240 & 3{,}057 \\
Reject           & 1{,}385 & 22{,}628 \\
\bottomrule
\end{tabular}
\end{table}

Sample size is small ($n=8$) because each hold test requires multi-step setup (envelope creation, hold trigger, human-response simulation); the table reports medians and maxima rather than percentiles. End-to-end mechanical latency (trigger to approve completion): $\sim$3.0\,ms. Approve is $\sim$2$\times$ slower than reject because it re-submits the original action through the full pipeline and writes two events (action result + hold response), whereas reject only settles commitment cost and writes one event.

\subsection{Comparison with AIOS}

Table~\ref{tab:aios} provides a qualitative capability comparison between \PunkGo{} and AIOS~\cite{aios}, the leading agent operating system. Direct latency comparison is not meaningful: AIOS includes LLM inference and tool execution in its critical path, whereas \PunkGo{}'s kernel measures only the commit path (validate through append). The two systems occupy different layers of the stack.

\begin{table}[t]
\centering
\caption{Qualitative comparison with AIOS~\cite{aios}.}
\label{tab:aios}
\small
\begin{tabularx}{\columnwidth}{lXX}
\toprule
\textbf{Dimension} & \textbf{AIOS} & \textbf{PunkGo} \\
\midrule
TCB & LLM (non-det.) & Kernel (det.) \\
Audit log & None & RFC 6962 Merkle \\
Inclusion proof & None & $O(\log n)$, 448\,B@10K \\
Consistency proof & None & $O(\log n)$ \\
Isolation & Permission hashmap & Capability, default deny \\
Governance & Token scheduling & Energy budget + envelope \\
Human approval & None & \texttt{hold\_on} + 20\% cost \\
Formal invariants & None & 5 (proof sketches) \\
Language & Python & Rust \\
\bottomrule
\end{tabularx}
\end{table}

\section{Related Work}
\label{sec:related}

\paragraph{AI Agent Governance.}
Agent Contracts~\cite{agent-contracts} formalizes a contract framework $C = (I, O, S, R, T, \Phi, \Psi)$ for resource-bounded agents, achieving 90\% token savings. Contracts define what agents \emph{may} do; \PunkGo{} records what they \emph{actually did}---the two are complementary. Infrastructure for AI Agents~\cite{infra-agents} proposes agent identity and attribution infrastructure but targets inter-agent interaction, not user-side verification on personal hardware. OpenAI Governance Practices~\cite{openai-governance} enumerates seven desirable properties (readability, monitoring, attributability, interruptibility, etc.) that presuppose persistent recording infrastructure but do not specify one.

\paragraph{Agent Access Control.}
AgentSentry~\cite{agentsentry} generates minimal ephemeral permissions; AgentBound~\cite{agentbound} provides MCP access control. Both enforce policy at invocation time but neither produces persistent tamper-evident records of what was actually executed.

\paragraph{Transparent Logs and Verifiable Data Structures.}
Certificate Transparency~\cite{rfc6962} established the canonical Merkle tree append-only log paradigm. Crosby and Wallach~\cite{crosby-wallach} formalized tamper-evident log data structure semantics. A2M~\cite{a2m} embeds an append-only primitive in the TCB---the \PunkGo{} kernel serves an analogous architectural role but targets AI agents rather than Byzantine consensus. CONIKS~\cite{coniks} extended transparent logs to end-user key verification. Google's Trillian~\cite{trillian} provides engineering-scale validation of the paradigm across multiple applications. \PunkGo{} applies this transparent log paradigm to a new domain: OS-level AI agent action auditing on personal hardware.

\paragraph{OS-Level Audit Systems.}
The Linux Audit Framework (auditd)~\cite{linux-audit} provides kernel-level syscall logging with configurable filtering rules and is widely deployed for compliance auditing (e.g., CAPP/EAL4+). However, auditd operates at the syscall granularity, not the agent-action level: it records \texttt{open()}, \texttt{write()}, \texttt{execve()} but cannot attribute these syscalls to a specific agent's semantic action (e.g., ``agent X modified document Y under envelope Z''). It also provides no cryptographic tamper evidence---logs are plain-text files that a root user can silently edit. Windows Event Tracing (ETW) shares these limitations. \PunkGo{} operates at a higher semantic level (agent actions rather than syscalls) and provides Merkle-based tamper evidence that auditd lacks.

\paragraph{eBPF-Based Runtime Auditing.}
eBPF~\cite{ebpf} extends the Linux kernel with sandboxed programs that attach to syscall and kernel events, enabling flexible runtime tracing and policy enforcement without kernel modification. Falco~\cite{falco} applies eBPF to container runtime security; AgentCgroup~\cite{agentcgroup} applies eBPF and cgroups to characterize AI agent resource dynamics. These tools operate at the syscall granularity---they observe \texttt{open()}, \texttt{execve()}, and cgroup counters, but cannot attribute syscall sequences to agent-level semantic actions or produce cryptographic tamper evidence. eBPF programs must terminate in bounded steps and cannot block for asynchronous human decisions, precluding mechanisms like \PunkGo{}'s hold workflow. Furthermore, eBPF is Linux-specific; \PunkGo{}'s user-space design is cross-platform by construction. The two approaches are complementary: eBPF provides fine-grained OS-level resource traces that could cross-validate agent-declared actions at the semantic layer.

\paragraph{Blockchain-Based Audit.}
Permissioned blockchains such as Hyperledger Fabric~\cite{hyperledger-fabric} provide tamper-evident distributed ledgers with ordering-service consensus, achieving 3{,}500+ tx/sec in multi-peer deployments. However, blockchain-based audit imposes significant architectural overhead for single-machine AI agent logging: Fabric requires multi-node consensus (even in crash-fault-tolerant mode), introduces network latency, and demands infrastructure far exceeding a personal workstation. For the target scenario---personal AI agents on local hardware---the trust model requires protecting a single user's history from untrusted agents, not achieving consensus among mutually distrusting organizations. \PunkGo{} adopts the cryptographic structure (Merkle trees) without the consensus overhead.

\paragraph{TEE-Based Attestation.}
Graphene-SGX (now Gramine)~\cite{graphene-sgx} demonstrates that unmodified applications can run inside Intel SGX enclaves with modest overhead, providing confidentiality and integrity against a malicious OS. TEE-based execution is complementary to \PunkGo{}: running the kernel inside an SGX/TDX enclave would elevate tamper-evidence guarantees from Level~1 (untrusted agents) to Level~2 (privileged platform adversaries), as discussed in \S\ref{sec:threat}. This integration is listed as future work.

\paragraph{OS-Level Security and Formal Verification.}
seL4~\cite{sel4} proved OS kernel functional correctness and capability isolation---but seL4 proves what agents \emph{cannot} do; it does not record what they \emph{did}. Capsicum~\cite{capsicum} brought practical capability security to UNIX---but without history recording. CertiKOS~\cite{certikos} extends verified kernel construction to concurrent settings with compositional layered specifications. IronFleet~\cite{ironfleet}'s three-layer refinement verification methodology (TLA+-style specification $\to$ protocol layer $\to$ implementation) serves as the template for \PunkGo{}'s formalization approach. AgentCgroup~\cite{agentcgroup} first characterized OS-level resource dynamics of AI agents---but without a security model or audit mechanism. Verus~\cite{verus} provides a Rust verification framework and is the target tool for \PunkGo{}'s future machine-checked proofs.

\paragraph{Agent Operating Systems.}
AIOS~\cite{aios} is the leading agent OS, providing scheduling, context management, and basic access control with 2.1$\times$ faster agent execution. Its access control uses a permission-group hashmap---functional for resource management but without audit logging, Merkle evidence, or capability-based isolation. The companion vision paper~\cite{llm-as-os} positions the LLM as an OS-level scheduler and resource manager. The characterization of ``LLM as TCB'' in Table~\ref{tab:aios} reflects the observation that AIOS's security boundary does not extend below the LLM scheduling layer, not a claim that AIOS deliberately places security trust in LLMs. The two systems address complementary concerns: AIOS optimizes agent execution; \PunkGo{} provides verifiable action recording.

\paragraph{Human Oversight and AI Safety.}
The Off-Switch Game~\cite{off-switch} proves utility-maximizing agents are incentivized to disable human intervention. Corrigibility~\cite{corrigibility} formalizes cooperation with correction. The Oversight Game~\cite{oversight-game} models oversight as a dynamic Markov game. These works establish theoretical foundations but lack implementation infrastructure---specifically, a verifiable record of \emph{when} humans oversaw, \emph{what} they decided, and \emph{whether} agents complied. \PunkGo{}'s hold mechanism and its Merkle-logged events are precisely the infrastructure these theoretical models assume but do not specify.

\paragraph{Sovereignty and Local-First.}
Local-First Software~\cite{local-first} establishes data ownership principles. Solid~\cite{solid} provides data storage location control. Sovereign Agents~\cite{sovereign-agents} analyzes how agents maintain sovereignty from humans. \emph{This paper argues the question points in the wrong direction}: in the context of AI agents acting on behalf of humans, what requires protection is human sovereignty over agents. The Right to History extends local-first principles from data ownership to computational history ownership.

\paragraph{Data Provenance and Trust.}
Provenance in databases~\cite{data-provenance} provides formal foundations for tracking the origin and transformation of data. Trust and reputation models for multi-agent systems~\cite{spice} establish frameworks for evaluating agent behavior based on historical records. \PunkGo{} provides the infrastructure layer that such provenance and trust systems require: a tamper-evident, independently verifiable record of agent actions, upon which provenance queries and reputation assessments can be built.

\section{Discussion and Future Work}
\label{sec:discussion}

\paragraph{Limitations.}
(1)~The current design uses a single-node, single-committer architecture with no multi-writer concurrency. This is deliberate: single-point linearization yields a deterministic total-order log, avoiding consensus protocol complexity. Distributed extension requires introducing consensus (e.g., Raft) and re-proving invariants under multi-committer semantics.
(2)~Formal arguments are structured proof sketches, not machine-checked.
(3)~Merkle proof generation exhibits $O(n)$ time due to loading all hash nodes from storage before computing the $O(\log n)$ proof path. At 10{,}000 entries proof generation takes $\sim$19\,ms; an attempted 100{,}000-entry benchmark was impractical due to $O(n^2)$ cumulative fill cost (each commit triggers $O(n)$ hash loading). Paginated or incremental hash loading would restore $O(\log n)$ generation time and is the critical scaling optimization.
(4)~The threat model assumes kernel code is trusted.
(5)~Throughput of $\sim$400 actions/sec is sufficient for personal agent usage but is bottlenecked by per-commit storage I/O (event persistence and Merkle checkpoint); batched commits or write-ahead buffering could improve throughput for higher-volume scenarios.
(6)~Hardware compute detection (INT8 TOPS) is currently configuration-file only; runtime probing is not implemented.
(7)~Execution environment isolation is the actor's responsibility; the kernel provides recording but not execution containment. Reference sandbox implementations for actors are future work.

\paragraph{Privacy Considerations and the Right to Erasure.}
An append-only action log creates an inherent tension with the right to erasure under data protection regulations~\cite{gdpr}. Three mitigating factors apply. First, in the primary deployment scenario---personal hardware, single user---the data controller and data subject are the same person; self-processing falls outside the material scope of most data protection frameworks (e.g., GDPR's household exemption). Second, for organizational deployments, the architecture supports \emph{crypto-erasure}: log entries can be encrypted with per-entity keys, and key destruction renders ciphertext unrecoverable while preserving the Merkle tree's structural integrity (hash nodes commit to ciphertext, not plaintext). Implementation of crypto-erasure is future work. Third, for high-risk AI systems, regulatory logging mandates (e.g., EU AI Act Article~12~\cite{eu-ai-act}) may create a legal basis that overrides erasure requests. A complete legal analysis of this interplay is beyond the scope of this paper.

\paragraph{Trust Boundary of Action Recording.}
The kernel records \emph{what agents claim to have done}, not necessarily \emph{what actually happened in the external world}. For the \texttt{execute} action type, the actor performs the execution externally and submits a structured payload (input/output OIDs, exit code, artifact hash). The kernel validates payload format and OID syntax, but does not verify whether OID references point to actual content or whether declared results correspond to actual execution outcomes. This is an explicit design choice, analogous to Certificate Transparency: CT log servers record that a Certificate Authority \emph{claims} to have issued a certificate, without verifying the certificate's content correctness~\cite{rfc6962}. The value lies in the tamper-evident \emph{attribution record}---who claimed what, when, under what authorization---which supports after-the-fact audit and dispute resolution. For use cases requiring stronger guarantees (e.g., verifying that an agent's declared file modification matches the actual file state), content verification at the actor's execution environment can be added as an extension without modifying the kernel.

\paragraph{Future Work.}
(1)~\textbf{Machine-checked verification}: End-to-end proofs from TLA+ specification to Rust code using Verus~\cite{verus}.
(2)~\textbf{TEE integration}: Running the kernel inside Intel SGX/TDX to reduce trust assumptions.
(3)~\textbf{Multi-agent collaboration}: A2A protocol integration for multi-agent scenarios.
(4)~\textbf{External tool ecosystem}: Modern AI agent frameworks increasingly adopt standardized tool-lifecycle protocols (e.g., MCP) that allow agents to discover, invoke, and manage external tools automatically. Integrating such protocols would enable \PunkGo{} to record tool invocations as auditable \texttt{execute} events, bridging the gap between automated tool orchestration and verifiable action history.
(5)~\textbf{Cross-organization verification}: A public Merkle proof verification service analogous to CT log servers.
(6)~\textbf{Performance optimization}: Paginated or incremental hash loading for Merkle proof generation (eliminating $O(n)$ proof-time overhead) and batched commit modes for higher throughput.
(7)~\textbf{Actor-side execution isolation}: Reference sandbox implementations (Docker, Firecracker microVM) for actors to use before submitting execute results to the kernel.
(8)~\textbf{Syscall-level cross-validation}: On Linux, eBPF-based syscall tracing~\cite{ebpf} could cross-validate agent-declared actions against actual system call sequences, though this requires a semantic-to-syscall attribution model that is an independent research problem.
(9)~\textbf{Hardware probing}: Runtime INT8 TOPS auto-detection replacing configuration files.

\section{Conclusion}
\label{sec:conclusion}

As AI agents increasingly act on behalf of humans, a fundamental question emerges: who records what happened? This paper defines the \emph{Right to History}---the principle that individuals are entitled to a complete, verifiable, tamper-evident record of every AI agent action---and provides a formal framework to realize it. The paper argues five system invariants (Append-Only, Completeness, Integrity, Boundary Enforcement, and Energy Conservation) via structured proof sketches, forming a security chain from behavior isolation to verifiable auditing. The \PunkGo{} Kernel implements this framework as a Rust sovereignty kernel, unifying the cryptographic auditing of Certificate Transparency, the capability isolation spirit of seL4, and an energy governance model into one system. Evaluation confirms all invariants hold under adversarial testing, with sub-1.3\,ms action latency and logarithmic-size Merkle proofs.

The Right to History is not an optional feature layered atop an existing system; it is a foundational architectural commitment---the event log and its invariants determine the system's structure, not the reverse. \PunkGo{} aims to encourage the systems community to treat verifiable action recording as a first-class concern in AI agent infrastructure.

\paragraph{Data and Code Availability.}
The \PunkGo{} Kernel source code, benchmark scripts, and all raw benchmark data reported in this paper are available at \url{https://github.com/PunkGo/punkgo-kernel} under the MIT License. Project homepage: \url{https://punkgo.ai}.

\bibliographystyle{plain}

\end{document}